\documentclass[a4paper, 11pt]{article}

\usepackage{enumerate}
\usepackage{bm}
\usepackage{microtype}
\usepackage{amsfonts}
\usepackage{amsthm,amsmath,amsfonts, amssymb}
\usepackage[dvips]{graphics,color}
\usepackage{tikz}
\usepackage{graphicx}
\usepackage{caption}

\newtheorem{theorem}{Theorem}
\newtheorem{corollary}{Corollary}
\newtheorem{lemma}{Lemma}

\title{Tight Bounds for Online Coloring of Basic Graph Classes\thanks{Work supported by the European Research Council, Grant Agreement No.\ 691672, project APEG.}}
\author{Susanne Albers\thanks{Department of Computer Science, Technical University of Munich, Garching, Germany. \texttt{albers@in.tum.de}.} \and 
	Sebastian Schraink\thanks{Department of Computer Science, Technical University of Munich, Garching, Germany. \texttt{schraink@in.tum.de}.} }
\date{}

\begin{document}

\maketitle

\begin{abstract}
	We resolve a number of long-standing open problems in online graph coloring. More specifically, we develop tight
	lower bounds on the performance of online algorithms for fundamental graph classes. An important
	contribution is that our bounds also hold for randomized online algorithms, for which hardly any results 
	were known. Technically, we construct lower bounds for chordal graphs. The constructions then allow us to
	derive results on the performance of randomized online algorithms for the following further graph classes: 
	trees, planar, bipartite, inductive, bounded-treewidth and disk graphs. It shows that the best
	competitive ratio of both deterministic and randomized online algorithms is $\Theta(\log n)$, where $n$ is the number of vertices of a graph.
	Furthermore, we prove that this guarantee cannot be improved
	if an online algorithm has a lookahead of size $O(n/\log n)$ or access to a reordering buffer of size $n^{1-\epsilon}$,
	for any $0<\epsilon\leq 1$. A consequence of our results is that, for all of the above mentioned graph classes except bipartite graphs,
	the natural \emph{First Fit} coloring
	algorithm achieves an optimal performance, up to constant factors, among deterministic and randomized
	online algorithms.
\end{abstract}

\section{Introduction}
\emph{Online graph coloring} is a classical problem in graph theory and online computation. It has applications
in job scheduling, dynamic storage allocation and resource management in wireless networks \cite{K1,M1,N1}. 
A problem instance is defined by an undirected graph $G=(V,E)$, consisting of a vertex set $V$ and an edge set $E$. 
Let $|V|=n$. The vertices arrive one by one in a sequence $\sigma = v_1,\ldots, v_n$ that may be
determined by an adversary. Whenever a new vertex $v_t$ arrives, $1\leq t \leq n$, its edges to previous vertices 
$v_s$ with $s<t$ are revealed. An online algorithm $\mathcal{A}$ has to immediately assign a feasible color to $v_t$, i.e.\ a color
that is different from those assigned to the neighbors of $v_t$ presented so far. The goal is to minimize the total 
number of colors used. 

For a graph $G$, let $\mathcal{A}(G)$ be the number of colors used by $\mathcal{A}$. Let $\chi(G)$ be the chromatic number of $G$, which is the minimum
number of colors needed to color $G$ offline. An online algorithm $\mathcal{A}$ is \emph{$c$-competitive} if $\mathcal{A}(G) \leq c \cdot \chi(G)$
holds for every graph $G$~\cite{ST}. If $\mathcal{A}$ is a randomized algorithm, then $E[\mathcal{A}(G)]$ is the expected
number of colors used by $\mathcal{A}$. The algorithm is $c$-competitive against oblivious adversaries if 
$E[\mathcal{A}(G)] \leq c \cdot \chi(G)$ holds for every $G$~\cite{BB+}. An oblivious adversary, when determining $\sigma$,
does not know the outcome of the random choices made by $\mathcal{A}$. We always evaluate randomized online algorithms
against this type of adversary. When considering specific graph classes, for a deterministic or randomized algorithm,
the competitive factor of $c$ must hold for every graph from the given class. 

The framework defined above is the standard online one. It is also interesting to explore settings where an algorithm is
given more power. An online algorithm $\mathcal{A}$ has \emph{lookahead $l$} if, upon the arrival of vertex $v_t$, the algorithm 
also sees the next $l$ vertices $v_{t+1}, \ldots, v_{t+l}$ along with their adjacencies to vertices in 
$\{v_1,\ldots, v_{t+l}\}$. Alternatively, an algorithm might have a \emph{buffer of size $b$} in which vertices can
be stored temporarily. The requirement is that at the end of step~$t$ the algorithm must have colored at least
$t-b$ vertices. A buffer is more powerful than lookahead because it allows the algorithm to partially reorder the
input sequence and delay coloring decisions. 
The value of a buffer has recently been explored for a variety of online problems, see e.g.~\cite{AR,E+} and references therein. 

\textbf{Previous work:} 
For general graphs, the competitive ratios are high compared to the trivial upper bound of $n$. Lovasz, Saks and Trotter~\cite{LST}
developed a deterministic online algorithm that achieves a competitive factor of $O(n/\log^* n)$. Vishwanathan~\cite{V1} devised a randomized
algorithm that attains a competitiveness of $O(n/\sqrt{\log n})$. This bound was improved to 
$O(n/\log n)$ by Halldorsson~\cite{H1}. Halldorsson and Szegedy~\cite{HS} proved that the competitive ratio of any deterministic
online algorithm is $\Omega(n/\log^2n)$. This lower bound also holds for randomized algorithms.
Moreover, it holds if a randomized algorithm has a lookahead or a buffer of size~$O(\log^2 n)$~\cite{HS}.

There has also been considerable research interest in online coloring for various graph classes. An early and celebrated
result proved by Bean~\cite{B1} in 1976 is that, for trees, every deterministic online algorithm can be forced to use 
$\Omega(\log n)$ colors. The \emph{First Fit} algorithm colors every tree with $O(\log n)$ colors~\cite{GL}. The natural 
strategy \emph{First Fit} assigns the lowest-numbered feasible color to each incoming vertex. Since trees have a chromatic 
number of~2, the best competitive ratio achievable by deterministic online algorithms is $\Theta(\log n)$. 
For bipartite graphs, there also exists a deterministic online algorithm that uses $O(\log n)$ colors~\cite{LST}, implying
that the best competitiveness of deterministic strategies is again $\Theta(\log n)$. However, \emph{First Fit}
performs poorly, as there are bipartite graphs for which it requires $\Omega(n)$ colors.  Kierstead and Trotter~\cite{KT}
proved that, for interval graphs, the best competitive ratio of deterministic online algorithms is equal to~3.

A paper directly related to our work is by Irani~\cite{I1}. She examined $d$-inductive graphs, also referred to as
$d$-degenerate graphs. They are defined as the graphs which admit a numbering of the vertices such that each vertex is adjacent to at most $d$ higher-numbered vertices. 
Every planar graph is $5$-inductive and every chordal graph $G$ is ($\chi(G)-1$)-inductive. Irani~\cite{I1} proved that 
\emph{First Fit} colors every $d$-inductive graph with $O(d \cdot \log n)$ colors. Furthermore, for every deterministic online algorithm 
$\mathcal{A}$, there exist graphs such that $\mathcal{A}$ uses $\Omega(d \cdot \log n)$ colors~\cite{I1}. Since $d$-inductive graphs 
have a chromatic number of at most $d+1$, the best competitive ratio achieved by deterministic online algorithms 
is $\Omega(\log n)$. 
For planar graphs a tight bound of $\Theta(\log n)$ holds because trees are planar. However, it was an open problem if a tight
competitiveness of $\Theta(\log n)$ holds for general chordal graphs. In fact, Irani~\cite{I1} raised the question if, for every 
deterministic online algorithm $\mathcal{A}$ and every $d$, there exists a chordal graph with chromatic number $d$ such that 
$\mathcal{A}$ uses $\Omega (d \cdot \log n)$ colors. Finally, for $d$-inductive graphs, Irani~\cite{I1} analyzed deterministic online 
algorithms with lookahead~$l$ and showed that the best competitiveness is $\Theta(\min \{\log n, n/l\})$.
A lower bound of $\Omega(\log\log n)$ on the competitive ratio of randomized online algorithms for $d$-inductive
graphs was given by Leonardi and Vitaletti~\cite{LV}. 

We address two further graph classes. Downey and McCartin~\cite{DC} studied online coloring of bounded treewidth graphs. 
For an introduction to treewidth see~\cite{Bo}. For any graph of treewidth $d$, \emph{First Fit} uses $O(d\cdot \log n)$ colors.
This is a consequence of Irani's work~\cite{I1} because a graph of treewidth $d$ is $d$-inductive~\cite{DC,I1}. Downey and 
McCartin~\cite{DC} showed that, on graphs of treewidth $d$, \emph{First Fit} can be forced to use
$\Omega(\frac{d}{\log (d+1)}\log n)$ colors. Last but not least, a disk graph is the intersection graph of a set of disks in the Euclidean plane.
Each vertex represents a disk; two vertices are adjacent if the two corresponding disks intersect. Online
coloring of disk graphs has received quite some attention because it models frequency assignment problems in wireless
communication networks, see~\cite{E2} for a survey. The best competitiveness achieved by a deterministic online algorithm is
$\Theta(\min\{\log n, \log \rho\})$, where $\rho$ is the ratio of the largest to smallest disk radius~\cite{C+,E1}. 
The result relies on the common assumption that an online algorithm does not use the disk representation, when making coloring 
decisions~\cite{C+,E1,E2}. It has been repeatedly raised as an open problem if the bound of $\Theta(\min\{\log n, \log \rho\})$
can be improved using randomization~\cite{C+,E1,E2}.

Recent work on online graph coloring has studied scenarios where an online algorithm can query oracle information 
about future input~\cite{BH1,BH2}. Moreover, online coloring of hypergraphs has been explored~\cite{BC,BC2}. 

\textbf{Our Contribution:} 
In this paper we settle the performance of online coloring algorithms for fundamental and
widely studied graph classes. More precisely, we prove lower bounds on the performance of online algorithms. 
These bounds match the best upper bounds known in the literature. An important contribution is that our bounds
also hold for randomized online algorithms, for which very few results were known.

First, in Sections~\ref{sec:det} and~\ref{sec:rand} we investigate chordal graphs. They have been studied extensively, cf.\ textbook~\cite{W1}.
We remind the reader that a graph is chordal if every induced cycle with four or more vertices has a chord. 
For a chordal graph $G$, the chromatic number $\chi(G)$ is equal to the largest clique size $\omega(G)$. Interval graphs
are a subfamily of chordal graphs. Chordal graphs in turn are perfect graphs, for which the offline coloring, maximum clique
and independent set problems can be solved in polynomial time. 

In Section~\ref{sec:det} we examine deterministic
online coloring algorithms. We prove that, for every deterministic algorithm $\mathcal{A}$ and every integer $d\geq 2$, there exists
a family of chordal graphs $G$ with $\chi(G)=d$ such that $\mathcal{A}$ uses $\Omega(d\cdot \log n)$ colors. This resolves 
the open problem raised by Irani~\cite{I1}. In Section~\ref{sec:rand} we extend this result to randomized online 
algorithms. The statement is identical to the one for deterministic algorithms, except that a randomized online algorithm uses
an expected number of $\Omega(d\cdot \log n)$ colors. Although the result for randomized algorithms is more general,
we give proofs for both deterministic and randomized policies. Our lower bound construction for deterministic algorithms
exhibits an adversarial strategy for generating worst-case graphs. Given this strategy, we show how to define a 
probability distribution on graphs so that Yao's principle~\cite{Y} can be applied. \emph{First Fit} colors
every chordal graph $G$ with $\chi(G)=d$ using $O(d\cdot \log n)$ colors. Hence, the optimal competitiveness
of deterministic and randomized online algorithms is $\Theta(\log n)$. 

In Section~\ref{sec:classes} we derive lower bounds for further graph classes, focusing on randomized online 
algorithms. For $d=2$, our lower bound construction for chordal graphs generates trees. It follows that, for
any randomized online algorithm $\mathcal{A}$, there exists a family of trees such that $\mathcal{A}$ needs an expected number
of $\Omega(\log n)$ colors. This complements the fundamental and early result by Bean~\cite{B1} for
deterministic algorithms. To the best of our knowledge, no lower bound on the performance of randomized 
online coloring algorithms for trees was previously known. Recall that trees have a chromatic number of~2. 
Vishwanathan~\cite{V1} gave a lower bound of $\Omega(\log n)$ on the expected number of colors
used by randomized online algorithms for graphs of chromatic number~2, i.e.\ bipartite graphs. However, the graphs in his construction have cycles.
Thus, Vishwanathan's lower bound does not apply to trees. Obviously, trees are planar and bipartite. Hence, our result for trees
directly implies that every randomized online algorithm can be forced to use $\Omega(\log n)$ colors in expectation for graphs of these two 
classes. The lower bounds are tight because known deterministic online algorithms color trees, planar and bipartite graphs 
with $O(\log n)$ colors~\cite{GL,I1,LST}. 

Section~\ref{sec:classes} also addresses inductive and bounded-treewidth graphs. Since every chordal graph $G$ is
$(\chi(G)-1)$-inductive and has treewidth $\chi(G)-1$, we derive the following results. For every randomized
online algorithm $\mathcal{A}$ and every $d\geq 1$, there exists a family of $d$-inductive graphs such that $\mathcal{A}$ uses
$\Omega(d\cdot \log n)$ colors. The same statement holds for graphs of treewidth $d$. We further show that the statement also
holds for strongly chordal graphs with chromatic number $d$. A chordal graph is strongly chordal if every  cycle
of even  length consisting of at least six vertices has an odd chord, i.e.\ an edge connecting two vertices that have an odd distance 
from each other in the cycle~\cite{F}. \emph{First Fit} colors any $d$-inductive graph and any graph of treewidth $d$ using $O(d\cdot \log n)$ colors. 
We conclude that, for all the graph classes considered so far, $\Theta(\log n)$ is the best competitiveness
of deterministic and randomized online algorithms. Finally, in Section~\ref{sec:classes} we study disk graphs. 
We prove that, for $d=2$, every graph of the probability distribution defined in Section~\ref{sec:rand} translates
to a disk graph. We then show that, for every randomized online algorithm $\mathcal{A}$ that does not use the disk representation, 
there exists a family of disk graphs forcing $\mathcal{A}$ to use an expected number of $\Omega(\min\{\log n, \log \rho\})$ colors, 
where $\rho$ is again the ratio of the largest to smallest disk radius. Hence randomization does not improve the asymptotic performance
of online coloring algorithms for disk graphs, cf.~\cite{C+,E1,E2}.

In Section~\ref{sec:look} we explore the settings where an online algorithm has lookahead or is equipped with
a reordering buffer. We show that a lookahead of size $O(n/\log n)$ does not improve the asymptotic performance
of randomized online algorithms. We prove the result for chordal graphs and then derive analogous results for
all the other graph classes. Irani~\cite{I1} gave a similar result for deterministic algorithms, considering 
inductive graphs. As a final result of this paper we demonstrate that a reordering buffer of size $n^{1-\epsilon}$, for any
$0<\epsilon \leq 1$, does not yield an improvement in the asymptotic performance guarantees of deterministic 
online algorithms. Again, we develop the result for chordal graphs and derive corollaries for the other graph
classes.

\textbf{Our Proof Technique:} We devise a technique for proving lower bounds that is relatively simple; we view
this as a strength of our results. The main idea is to recursively construct trees of cliques, which in turn form forests. 
In a recursive step the construction combines forests by adding or not adding a new clique in a specific way. Our construction
resembles the one by Bean~\cite{B1} but differs in an important aspect that allows us to obtain lower bounds for randomized
algorithms. The construction by Bean builds a tree $T_k$, $k\in \mathbb{N}$, by joining trees $T_j$, for $j<k$, so that
any deterministic online algorithm must use a $k$-th new color for \emph{some} vertex of $T_k$. This vertex then becomes
the root of $T_k$. An oblivious adversary, playing against a randomized online algorithm, cannot identify with sufficiently 
high probability such vertices exhibiting a new color. Instead, our construction maintains the invariant that the
root vertices of each forest use a large number of colors, given any deterministic online algorithm. For randomized
algorithms, a corresponding invariant holds with probability of at least $1/2$. 

\textbf{Convention:} Unless otherwise stated, logarithms are base~2.

\section{Deterministic online algorithms for chordal graphs}\label{sec:det}
We establish a lower bound on the performance of any deterministic online coloring algorithm. 
\begin{theorem}\label{th:det}
	Let $d \in \mathbb{N}$ with $d\geq 2$ be arbitrary. For every deterministic online algorithm $\mathcal{A}$
	and every $n \in \mathbb{N}$ with $n\geq 2d^2$, there exists a $n$-vertex chordal graph $G$ with chromatic number 
	$\chi(G)=d$ such that $\mathcal{A}$ uses $\Omega(d\cdot \log n)$ colors to color $G$. 
\end{theorem}
The proof of Theorem~\ref{th:det} relies on Lemma~\ref{lem:det}, which we prove first.
\begin{lemma}\label{lem:det}
	Let $d \in \mathbb{N}$ with $d\geq 2$ be arbitrary. For every deterministic online algorithm $\mathcal{A}$ and 
	every $k\in \mathbb{N}$, there exists a chordal graph $G_k$ having chromatic number $\chi(G_k)=d$ and consisting of 
	$n_k\leq d2^k$ vertices such that $\mathcal{A}$ is forced to use at least $c_k\geq (d-1)k/4$ colors to color $G_k$. 
\end{lemma}
\begin{proof}
	We describe how an adversary constructs a chordal graph $G_k$, $k\in \mathbb{N}$. Such a graph is built up recursively and
	consists of graphs $G_j$, where $j<k$. We assume that $d$ is even. The construction of $G_k$ can be adapted easily if $d$ is 
	odd; details will be given later. On a high level $G_k$ is a forest, i.e. a collection of disjoint trees, each 
	having a distinguished root node. In every tree $T$ of $G_k$, each tree node represents a clique of size $d/2$ in $G_k$.
	If two tree nodes $u_T$ and $v_T$ are connected by a tree edge in $T$, then any two vertices $u\in u_T$
	and $v\in v_T$ are connected by an edge in $G_k$. Hence $u_T$ and $v_T$ form a clique of size $d$ in $G_k$. Since $G_k$
	is a forest, it consists of several connected components. One can add a final vertex and edges in order to connect the
	various trees; details will be given at the end of the proof. 
	
	We proceed with the concrete construction of $G_k$, for increasing values of $k\in \mathbb{N}$. As mentioned above, each tree $T$
	of $G_k$ has a distinguished root node consisting of $d/2$ vertices in $G_k$. Let $r(T)$ be the set of these $d/2$ vertices. 
	Moreover, let $r(G_k)$ be the union of these sets $r(T)$, taken over all $T$ of $G_k$. We refer to the elements of $r(G_k)$ as 
	the \emph{root vertices of $G_k$}. They are important because the online algorithm $\mathcal{A}$ will be forced to use a large number 
	of colors for $r(G_k)$. For any subset $V'$ of the vertices of $G_k$, let $\mathcal{C}_\mathcal{A}(V')$ be the set of colors used by $\mathcal{A}$ 
	to color $V'$. 
	
	The strategy of the adversary to generate a graph $G_k$ is adaptive, i.e.\ the exact structure of the graph depends 
	on the coloring decisions of $\mathcal{A}$. Nevertheless, during the bottom-up construction of $G_k$, for increasing $k\in \mathbb{N}$, 
	the following invariants will be maintained.
	\begin{enumerate}[(1)]
		\item Algorithm $\mathcal{A}$ uses at least $\frac{d}{4} \cdot k$ colors for the root vertices of $G_k$, i.e.\
		$\left| \mathcal{C}_\mathcal{A}\left( r(G_k) \right) \right| \geq \frac{d}{4} \cdot k$. 
		\item $G_k$ is a union of connected components, each of which can be represented by a tree $T$. Each tree node 
		is a clique of size $d/2$. Every tree $T$ has a distinguished root node containing a set $r(T)$ of $d/2$ root vertices in $G_k$.
		\item $G_k$ is chordal.
		\item The maximum clique size is $\omega(G_k) = d$.
		\item The number of vertices satisfies $n_k \leq \frac{d}{2} \cdot (2^{k+1} -1)$.
	\end{enumerate}
	Invariants (3) and (4) together imply that $\chi(G_k) =\omega(G_k) = d$ holds.
	In invariant~(1) and the following technical
	exposition integer values are compared to expressions of the form $\frac{d}{4} \cdot k$, which might not be integer.
	We remark that the statements, comparisons and calculations hold without considering the rounded expressions. 
	
	\textbf{Construction of the base graph \bm{$G_1$}}: $G_1$ is a clique of size $d$. The adversary may present the corresponding
	vertices in an arbitrary order. The set of root vertices $r(G_1)$ is an arbitrary subset $R$ of size $d/2$ of the vertices
	of $G_1$. The remaining $d/2$ vertices form a second tree node. The resulting tree $T$ is depicted in Figure~\ref{fig:G_1}. 
	We can easily verify properties~(1--5).\\
	(1) Since $R=r(G_1)$ is a clique of size $d/2$, $\mathcal{A}$ uses $d/2$ colors for it, i.e.\ 
	$\left|\mathcal{C}_\mathcal{A}(r(G_1))\right| \geq \frac{d}{4}$.\\
	(2) $G_1$ consists of one connected component which represents a tree, as described above and shown in Figure~\ref{fig:G_1}.\\
	(3) $G_1$ is a clique and thus chordal.\\
	(4) The maximum clique size $\omega(G_1)$ is exactly $d$.\\
	(5) There holds $n_1 = d \leq {\frac{3}{2}}\cdot d = \frac{d}{2} \cdot (2^{1+1} -1)$.
	
\begin{figure}[ht]
	\begin{minipage}[t]{7.2cm}
		\begin{center}
			\includegraphics[width=0.03\textwidth]{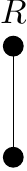}
		\end{center}
		\captionsetup{justification=centering}
		\caption{The tree $T$ representing $G_1$}\label{fig:G_1}
	\end{minipage}
	\begin{minipage}[t]{9.2cm}
		\begin{center}
			\includegraphics[width=0.65\textwidth]{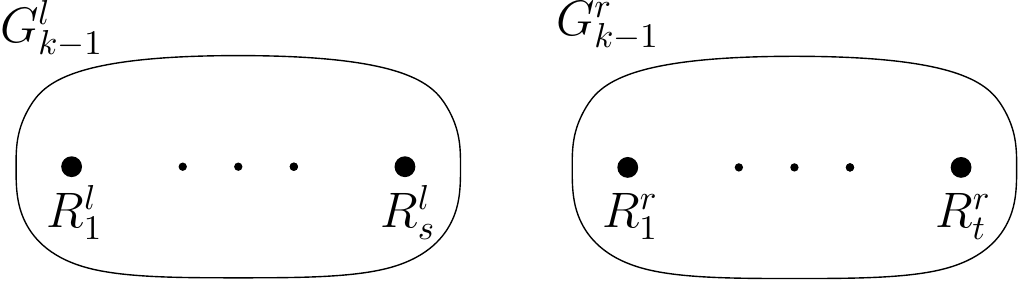}
		\end{center}
		\captionsetup{justification=centering}
		\caption{The general structure of $G^l_{k-1}$ and\\$G^r_{k-1}$ restricted to the root vertices}
		\label{fig:G_2}
	\end{minipage}
\end{figure}
	
	\textbf{Construction of the graph \bm{$G_k$}, \bm{$k>1$}}: Assume that the adversary can generate graphs $G_j$, for any $j<k$, satisfying
	invariants~(1--5). The construction of $G_k$ proceeds as follows. First the adversary recursively generates two independent
	graphs of type $G_{k-1}$, i.e.\ it twice executes the strategy for generating a graph $G_{k-1}$. Let $G^l_{k-1}$ and 
	$G^r_{k-1}$ be these two graphs. They are created one after the other.  
	We remark that $G^l_{k-1}$ and $G^r_{k-1}$ need not be identical
	because $\mathcal{A}$'s coloring decision in one graph can affect its decisions in the other one. 
	
	In the following we focus on the root vertices of $G^l_{k-1}$ and $G^r_{k-1}$. In particular, we consider the colors used by $\mathcal{A}$. 
	Invariant~(1) implies that $\left|\mathcal{C}_\mathcal{A}(r(G^l_{k-1}))\right| \geq \frac{d}{4}(k-1)$ and 
	$\left|\mathcal{C}_\mathcal{A}(r(G^r_{k-1}))\right| \geq \frac{d}{4}(k-1)$.
	We distinguish two cases depending on the total number of colors used, i.e.\ the cardinality of 
	$\mathcal{C}_\mathcal{A}(r(G_{k-1}^l) \cup r(G_{k-1}^r))$. To this end we introduce some notation. Assume that $G_{k-1}^l$ consists
	of $s$ connected components, which we number in an arbitrary way. Each component/tree $T^l_i$ has a distinguished root 
	containing a set $r(T^l_i)$ of $d/2$ root vertices. We abbreviate $R_i^l = r(T^l_i)$, $1\leq i \leq s$. Similarly, assume that 
	$G_{k-1}^r$ consists of $t$ connected components. Set $r(T^r_j)$ is the set of root vertices in the component $T^r_j$. 
	Let $R_j^l = r(T^r_j)$, $1\leq j \leq t$. There holds $r(G_{k-1}^l)= \bigcup_{i=1}^s R_i^l$ and $r(G_{k-1}^r)= \bigcup_{j=1}^t R_j^r$.
	Figure~\ref{fig:G_2} shows the general structure of $G^l_{k-1}$ and $G^r_{k-1}$ by focusing on the roots. The left-hand side
	of the figure depicts $G^l_{k-1}$ as a union of connected components rooted at $R_1^l, \ldots, R_s^l$, respectively. The right-hand
	side shows $G^r_{k-1}$ as a collection of components rooted at $R_1^r, \ldots, R_s^r$.
	\smallskip
	
	\textbf{\emph{Case 1:}}
	Assume that $\left| \mathcal{C}_\mathcal{A}(r(G_{k-1}^l) \cup r(G_{k-1}^r)) \right| \geq \frac{d}{4} \cdot k$. In this
	case the adversary defines $G_k$ as the union of $G^l_{k-1}$ and $G^r_{k-1}$. No further vertices or edges are added. It is easy
	to verify the five invariants because $G^l_{k-1}$ and $G^r_{k-1}$ satisfy them by inductive assumption.\\
	(1) The condition of Case~1 ensures $\left|\mathcal{C}_\mathcal{A}(r(G_k)) \right| = \left| \mathcal{C}_\mathcal{A}(r(G_{k-1}^l) \cup r(G_{k-1}^r))\right| 
	\geq \frac{d}{4} \cdot k$.\\
	(2) The invariant is satisfied since $G_k$ is the union of $G^l_k$ and $G^r_k$. \\
	(3) $G_k$ is chordal because $G^l_k$ and $G^r_k$ are, and no further vertices or edges have been added.\\
	(4) There holds $\omega(G_k)=d$, as $\omega(G^l_{k-1})=\omega(G^r_{k-1})=d$.\\
	(5) Let $n^l_{k-1}$ and $n^r_{k-1}$ be the number of vertices in $G^l_{k-1}$ and $G^r_{k-1}$, respectively. There holds
	$n_k = n^l_{k-1}+n^r_{k-1} \leq 2 \cdot (\frac{d}{2} \cdot (2^{k}-1)) = \frac{d}{2} \cdot (2^{k+1}-2) 
	\leq \frac{d}{2} \cdot (2^{k+1}-1)$. The first inequality follows because (5)~holds for $n^l_{k-1}$ and $n^r_{k-1}$.
	\smallskip
	
	\textbf{\emph{Case 2:}} Next assume that $\left| \mathcal{C}_A(r(G_{k-1}^l) \cup r(G_{k-1}^r)) \right| < \frac{d}{4} \cdot k$.
	In this case the adversary adds a set $R$ of $d/2$ vertices that form a clique. Moreover, for every vertex of $R$ there is an 
	edge to every vertex in $R_i^l$, for $i=1,\ldots, s$. In other words, every vertex of $R$ has edges to all root vertices
	of $r(G_{k-1}^l)$. The vertices of $R$ together with their adjacent edges may be presented by the adversary
	in an arbitrary order. The resulting structure is depicted in Figure~\ref{fig:G_3}. Set $R$ and the connected components of $G^l_{k-1}$ 
	rooted at $R_1^l, \ldots, R_s^l$ form a single component rooted at $R$. There is a tree edge between $R$ and every $R_i^l$, $1\leq i \leq s$. 
	The newly created component forms a tree rooted at $R$ because the components of $G_{k-1}^l$ represent trees rooted at 
	$R_1^l, \ldots, R_s^l$. Graph $G_k$ is the union of the new component and the components of $G_{k-1}^r$. The set of
	root vertices of $G_k$ consists of $R$ and the root vertices of $G_{k-1}^r$. Formally, $r(G_k) = R\cup R_1^r\cup \ldots, \cup R_t^r$. 
	It remains to verify the five invariants.
	
	\begin{figure}[h]
		\centering
		\includegraphics[width=0.40\textwidth]{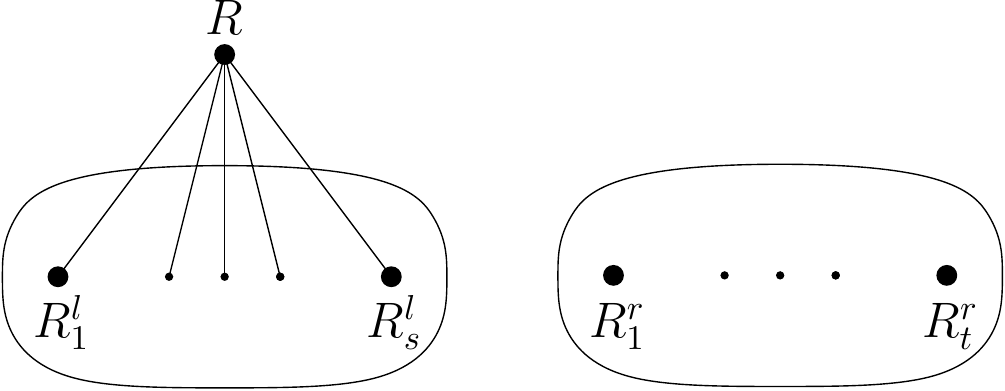}
		\captionsetup{justification=centering}
		\caption{The graph $G_k$ with the new addition of $R$}\label{fig:G_3}
	\end{figure}
	
	\noindent
	(1) We analyze the number of colors that $\mathcal{A}$ uses for the root vertices in $G_k$. In a first step, among the colors 
	$\mathcal{C}_\mathcal{A}(r(G_{k-1}^l)) \cup \mathcal{C}_\mathcal{A}(r(G_{k-1}^r))$ for the roots of $G_{k-1}^l$ and $G_{k-1}^r$, we upper bound the number $q$ of colors occurring in ${\cal C}_\mathcal{A}(r(G_{k-1}^r))$ only. By assumption 
	$\left|\mathcal{C}_\mathcal{A}(r(G_{k-1}^l)) \cup \mathcal{C}_\mathcal{A}(r(G_{k-1}^r)) \right| = 
	\left| \mathcal{C}_\mathcal{A}(r(G_{k-1}^l) \cup r(G_{k-1}^r)) \right| < \frac{d}{4} \cdot k$. 
	There holds $\mathcal{C}_\mathcal{A}(r(G_{k-1}^l)) 
	\geq \frac{d}{4}(k-1)$. We obtain $q = 
	\linebreak 
	\left|{\cal C}_\mathcal{A}(r(G_{k-1}^r))\setminus{\cal C}_\mathcal{A}(r(G_{k-1}^l))\right| = 
	\left|{\cal C}_\mathcal{A}(r(G_{k-1}^r)) \cup {\cal C}_\mathcal{A}(r(G_{k-1}^l))\right| - \left|{\cal C}_\mathcal{A}(r(G_{k-1}^l))\right| < \frac{d}{4}$. 
	Next consider the vertices in $R$. We upper bound the number of colors from ${\cal C}_\mathcal{A}(r(G_{k-1}^r))$
	that $\mathcal{A}$ can use for $R$. Observe that ${\cal C}_\mathcal{A}(r(G_{k-1}^r))$ is the disjoint union of 
	$\mathcal{C}_\mathcal{A}(r(G_{k-1}^l))\cap\mathcal{C}_\mathcal{A}(r(G_{k-1}^r))$ and $\mathcal{C}_\mathcal{A}(r(G_{k-1}^r))\setminus \mathcal{C}_\mathcal{A}(r(G_{k-1}^l))$.
	Every vertex of $R$ is adjacent to every vertex in $r(G_{k-1}^l)$. Hence,
	$\mathcal{A}$ cannot apply a color occurring in $\mathcal{C}_\mathcal{A}(r(G_{k-1}^r))\cap \mathcal{C}_\mathcal{A}(r(G_{k-1}^l))$ to a vertex in $R$.
	Only a color of $\mathcal{C}_\mathcal{A}(r(G_{k-1}^r))\setminus \mathcal{C}_\mathcal{A}(r(G_{k-1}^l))$ is feasible, and the latter set has cardinality $q<d/4$. 
	Since $R$ is a clique of size $d/2$ algorithm $\mathcal{A}$ must use at least $d/2-q >d/4$ colors not contained in $\mathcal{C}_\mathcal{A}(r(G_{k-1}^r))$ 
	to color the vertices of $R$. As $r(G_k) = R\cup r(G_{k-1}^r)$, we conclude $\left|\mathcal{C}_\mathcal{A}(r(G_k))\right| = 
	\left|\mathcal{C}_\mathcal{A}(R\cup r(G_{k-1}^r))\right| = \left|\mathcal{C}_\mathcal{A}(r(G_{k-1}^r))\right| +\left|\mathcal{C}_\mathcal{A}(R)\setminus \mathcal{C}_\mathcal{A}(r(G_{k-1}^r))\right| 
	\geq \frac{d}{4}(k-1) +\frac{d}{4}= \frac{d}{4}k$.\\
	(2) By construction $G_k$ is a collection of connected components, forming trees rooted at $R$ and $R_1^r, \ldots, R_t^r$,
	respectively.\\ 
	(3) In $G_k$ consider a simple cycle $C$ with at least four vertices and assume that at least one vertex is in $R$. If three or
	more vertices of $C$ are in $R$, then there is a chord because $R$ is a clique. If $C$ contains one or two vertices of $R$, then
	$C$ can visit only one connected component of $G^l_{k-1}$. Suppose that it visits the one rooted at $R_i^l$. Cycle
	$C$ must contain two vertices of $R_i^l$. Each of these two vertices has an edge to every vertex of $R$ in $C$.
	Hence $C$ has a chord. Since $G_{k-1}^l$ and $G_{k-1}^r$, and thus the components rooted at $R_1^l, \ldots, R_s^l$ and
	$R_1^r, \ldots, R_t^r$, are chordal, so is $G_k$. \\
	(4) Set $R$ and each $R_i^l$, $1\leq i \leq s$, form a clique of size $d$. The vertices of $R$ are not connected to
	any vertices outside $R_i^l$, $1\leq i \leq s$. Hence no other cliques are formed by the addition of $R$. Since
	$\omega(G^l_{k-1})=\omega(G^r_{k-1})=d$ it follows $\omega(G_k)=d$.\\
	(5) Again, let $n^l_{k-1}$ and $n^r_{k-1}$ be the number of vertices in $G^l_{k-1}$ and $G^l_{k-1}$. We have $n_k = n^l_{k-1}+n^r_{k-1} + \frac{d}{2} \leq 2 \cdot (\frac{d}{2} \cdot (2^{k}-1)) + \frac{d}{2} = 
	\frac{d}{2} \cdot (2^{k+1}-2) + \frac{d}{2} = \frac{d}{2} \cdot (2^{k+1}-1)$. \\
	The construction and analysis of $G_k$ is complete.
	
	Graph $G_k$ consists of several connected components if $k>1$. The adversary can create a connected graph by adding a final vertex 
	$v_f$ that has an edge to exactly one root vertex in each of the components. The resulting graph remains chordal because there is no simple cycle 
	containing $v_f$. By the addition of $v_f$ the maximum clique size does not change. Including $v_f$ the total number
	of vertices is upper bounded by $\frac{d}{2}(2^{k+1}-1) +1 \leq d2^k$ because $d\geq 2$. The lemma follows from invariants~(1)
	and (3--5) because $\chi(G_k)= \omega(G_k) =d$. 
	
	We finally address the case that $d$ is odd. In this case the adversary executes the graph construction described above
	for parameter $d-1$, which is even. In the end when $G_k$ is generated for the desired $k$, the adversary adds a final
	vertex to each base graph $G_1$. This vertex has edges to every other vertex of the corresponding $G_1$. This increases 
	the maximum clique size from $d-1$ to $d$. The new graph remains chordal. The number of colors used by algorithm $\mathcal{A}$
	is at at least $\frac{d-1}{4} k$. We observe that the number of base graphs $G_1$ in $G_k$ is $2^{k-1}$. Hence, in the extended graph 
	the total number of vertices is upper bounded by 
	$\frac{d-1}{2} (2^{k+1}-1) + 2^{k-1} \leq \frac{d}{2} (2^{k+1}-1)$. If $k>1$, the adversary can add a final vertex to link 
	the various components. Again the lemma follows. 
\end{proof}
\begin{proof}[Proof of Theorem~\ref{th:det}] 
	Given $d$ and $n$, let $k=\lfloor \log(n/d)\rfloor$. There holds $k\in \mathbb{N}$ because $n\geq 2d^2 > 2d$. For every deterministic
	online algorithm, by Lemma~\ref{lem:det}, there exists a chordal graph $G_k$ with chromatic number $\chi(G_k)=d$ such that
	$\mathcal{A}$ uses at least $c_k\geq (d-1)k/4$ colors. Graph $G_k$ has $n_k \leq d2^k$ vertices. By the choice of $k=\lfloor \log(n/d)\rfloor$,
	we have $n_k\leq n$. To $G_k$ we add $n-n_k$ vertices, all of which have one edge to an arbitrary vertex of $G_k$. The resulting
	$n$-vertex graph remains chordal and $\chi(G)=d$. Since $d\geq2$, there holds $c_k\geq dk/8$. We have $k\geq \log n-\log d-1$. 
	Inequality $n\geq 2d^2$ is equivalent to $d\leq \sqrt{n/2}$. Thus, $k\geq \log(n/2)-1/2\cdot \log(n/2) = 1/2\cdot \log (n/2)$. 
	As $n\geq 2d^2\geq 4$, there holds $\log(n/2)\geq 1/2\cdot \log n$. Hence, the number of colors used by
	$\mathcal{A}$ is at least $c_k\geq d\log n /32$.
\end{proof} 
In Theorem~\ref{th:det} the lower bound on $n$ can be reduced from $2d^2$ to $2d^{1+\epsilon}$, for any $0<\epsilon<1$. Then the number of colors used by $\mathcal{A}$ is $\Omega(\epsilon \cdot d \cdot \log n)$. 

\section{Randomized online algorithms for chordal graphs}\label{sec:rand}
We extend the result of Theorem~\ref{th:det} to randomized algorithms against oblivious adversaries. 
\begin{theorem}\label{th:rand}
	Let $d \in \mathbb{N}$ with $d\geq 2$ be arbitrary. For every randomized online algorithm $\mathcal{A}$ and every
	$n \in \mathbb{N}$ with $n\geq 12d^2$, there exists a $n$-vertex chordal graph $G$ with chromatic number $\chi(G)=d$, 
	presented by an oblivious adversary, such that the expected number of colors used by $\mathcal{A}$ to color $G$ is $\Omega(d\cdot \log n)$. 
\end{theorem}
In order to prove Theorem~\ref{th:rand} we resort to Yao's principle~\cite{Y} and show the following Lemma~\ref{lem:rand}.
\begin{lemma}\label{lem:rand}
	Let $d \in \mathbb{N}$ with $d\geq 2$ be arbitrary. For every $k\in \mathbb{N}$, there exists a probability 
	distribution on a set ${\cal G}_k$ of chordal graphs with the following properties. For every $G_k \in {\cal G}_k$, $\chi(G_k)=d$ 
	and the number of vertices is at most $d\cdot 12^k$. The expected number of colors used by any 
	deterministic online algorithm to color a graph drawn according to the distribution is at least $(d-1)k/8$.
\end{lemma}
\begin{proof}
	For every $k\in \mathbb{N}$ we define a set ${\cal G}_k$ of chordal graphs $G_k$, each having a chromatic number of $d$. Moreover,
	we specify the order in which the vertices of any $G_k\in {\cal G}_k$ are presented to a deterministic online algorithm $\mathcal{A}$. 
	The distribution on ${\cal G}_k$ is the uniform one, i.e.\ each $G_k\in {\cal G}_k$ is chosen with the same probability.
	We assume that $d$ is even. The definition of ${\cal G}_k$ can be adapted easily if $d$ is odd; details are given
	at the end of the proof. 
	
	The set ${\cal G}_k$ is built recursively based on ${\cal G}_{k-1}$. The construction of graphs $G_k\in {\cal G}_k$ is a generalization 
	of the one presented in the proof of Lemma~\ref{lem:det}. A major difference is that any $G_k\in {\cal G}_k$ contains twelve graphs of
	${\cal G}_{k-1}$, which are grouped into six pairs. For each pair a clique of size $d/2$ may or may not be added. As before, every
	$G_k\in {\cal G}_k$ is a union of connected components. Each such component can be represented by a tree with a distinguished root
	vertex. Every tree vertex is a set of $d/2$ vertices forming a clique in $G_k$. We reuse the notation of the proof of 
	Lemma~\ref{lem:det}. Given $G_k\in {\cal G}_k$, for any component/tree $T$ of $G_k$, $r(T)$ is the set of $d/2$ vertices in 
	the root of $T$. Set $r(G_k)$ is the union of all $r(T)$, taken over all $T$ of $G_k$. Finally ${\cal C}_\mathcal{A}(r(G_k))$ is the
	set of colors used by $\mathcal{A}$ for the vertices of $r(G_k)$. 
	
	During the recursive construction of ${\cal G}_k$, for increasing $k\in \mathbb{N}$, the following invariants are maintained. 
	Compared to the proof of Lemma~\ref{lem:det}, (1)~and (5) differ. Invariant~(1) states that, for a randomly chosen $G_k$, 
	every deterministic online algorithm needs, with probability greater than 1/2, at least
	$dk/4$ colors for the root vertices $r(G_k)$.
	Invariant~(5) gives an adjusted bound on the size of any $G_k$. 
	
	\begin{enumerate}[(1)]
		\item If $G_k$ is chosen uniformly at random from ${\cal G}_k$, then for any deterministic online algorithm $\mathcal{A}$,
		$\Pr[\left| \mathcal{C}_\mathcal{A}\left(r(G_k) \right) \right| \geq dk/4] > 1/2$. This holds independently of other connected
		components $\mathcal{A}$ might have already colored.
		\item Every $G_k\in {\cal G}_k$ is a union of connected components, each of which can be represented by a tree $T$. Each tree node 
		is a clique of size $d/2$. Every tree $T$ has a distinguished root containing a set $r(T)$ of $d/2$ root vertices in $G_k$.
		\item Every $G_k\in {\cal G}_k$ is chordal.
		\item For every $G_k\in {\cal G}_k$, the maximum clique size is $\omega(G_k) = d$.
		\item For every $G_k\in {\cal G}_k$, the number $n_k$ of vertices satisfies $n_k \leq  d(12^k-1) $.
	\end{enumerate}
	
	\textbf{Graph set \bm{${\cal G}_1$}}:
	The set only contains $G_1$, the base graph used in the proof of Lemma~\ref{lem:det}, which is a 
	clique of size $d$. The vertices of $G_1$ may be presented in any order to a deterministic online algorithm. Again, the set $r(G_1)$ of 
	root vertices is an arbitrary subset of size $d/2$ of the vertices of $G_1$. The remaining $d/2$ vertices form a second tree node. 
	Every deterministic online algorithm, with probability~1, needs $d/2$ colors for $r(G_1)$, which implies~(1). Invariants (2--4) are obvious. 
	As for~(5), there holds $n_1= d \leq d(12-1)$.
	\smallskip
	
	\textbf{Graph set \bm{${\cal G}_k$}, \bm{$k>1$}}: Assume that the set ${\cal G}_{k-1}$ satisfying (1--5) has been constructed. First, in order
	to build ${\cal G}_k$, all possible $12$-tuples of graphs of ${\cal G}_{k-1}$ are formed. In assigning tuple entries, 
	graphs of ${\cal G}_{k-1}$ are selected with replacement. Hence, a total of $\left|{\cal G}_{k-1}\right|^{12}$ tuples are built.
	For each tuple, $2^6$ graphs are added to ${\cal G}_k$ in the following way. Let $\tau$ be any fixed tuple. Six graph
	pairs are formed. For $i=1,\ldots, 6$, let $G^{i,l}_{k-1}$ and $G^{i,r}_{k-1}$ be the graphs in tuple entries $2i-1$ and $2i$,
	respectively. To the $i$-th pair a clique $R_i$ of size $d/2$ may or may not be added. The possible additions, over the six pairs,
	can be represented by a bit vector $\vec{b}=(b_1,\ldots,b_6)$. More specifically, given $\tau$ and any such bit vector $\vec{b}$,
	a graph $G_k$ is constructed as follows. For $i=1,\ldots, 6$, a subgraph $G^i_k$ is generated. If $b_i=0$, then $G^i_k$ is the union of 
	$G^{i,l}_{k-1}$ and $G^{i,r}_{k-1}$. The set $r(G^i_k)$ of root vertices is the union of $r(G^{i,l}_k)$ and $r(G^{i,r}_k)$. 
	If $b_i=1$, then a clique $R_i$ of size $d/2$ is added to $G^{i,l}_{k-1}$ and $G^{i,r}_{k-1}$. Every vertex of $R_i$ has an edge
	to every vertex of $r(G^{i,l}_{k-1})$. Subgraph $G^i_{k}$ consists of the newly created component rooted at $R_i$ and $r(G^{i,r}_{k-1})$, i.e. $r(G^i_{k-1}) = R_i\cup r(G^{i,r}_{k-1})$. Graph $G_k$ is the union of the $G^i_k$ and the set $r(G_k)$ is the union of the 
	$r(G^i_k)$, $1\leq i\leq 6$. When $G_k$ is presented to $\mathcal{A}$
	, the subgraphs $G^i_k$ are revealed one by one, 
	$1\leq i\leq 6$. For each $G^i_k$ the graphs $G^{i,l}_{k-1}$ and $G^{i,r}_{k-1}$ are presented recursively. Finally, the vertices 
	of $R_i$, if they exist, are shown. It remains to verify the invariants.
	
	(1) Let $G_k$ be a graph drawn uniformly at random from ${\cal G}_k$. Consider any subgraph $G^i_k$, $1\leq i\leq 6$, containing
	$G^{i,l}_k$ and $G^{i,r}_k$. By the construction of ${\cal G}_k$, both $G^{i,l}_k$ and $G^{i,r}_k$ represent graphs drawn 
	uniformly at random from ${\cal G}_{k-1}$. Let $\mathcal{A}$ be any deterministic online algorithm. Invariant~(1) for $k-1$ implies 
	$\Pr[|{\cal C}_{\mathcal{A}}(r(G^{i,l}_{k-1}))| \geq d(k-1)/4] > 1/2$ and $\Pr[|{\cal C}_{\mathcal{A}}(r(G^{i,r}_{k-1}))|\geq d(k-1)/4] > 1/2$. 
	Moreover it implies $\Pr[|{\cal C}_{\mathcal{A}}(r(G^{i,l}_{k-1}))| \geq d(k-1)/4\ \mbox{and}\  |{\cal C}_{\mathcal{A}}(r(G^{i,r}_{k-1}))| \geq d(k-1)/4] > 1/4$.
	Let ${\cal E}^i$ be the latter event that $|{\cal C}_{\mathcal{A}}(r(G^{i,l}_{k-1}))| \geq d(k-1)/4$ and 
	$|{\cal C}_{\mathcal{A}}(r(G^{i,r}_{k-1}))|\geq d(k-1)/4$ hold.
	
	Assume that ${\cal E}^i$  holds. There are two cases, which correspond to those analyzed in the proof of Lemma~\ref{lem:det}. If
	$|{\cal C}_{\mathcal{A}}(r(G^{i,l}_{k-1}) \cup r(G^{i,l}_{k-1}))|\geq dk/4$, then $|{\cal C}_\mathcal{A}(r(G^i_k))| \geq dk/4$ if $R_i$ is not 
	added to $G^{i,l}_k$ and $G^{i,r}_k$, which happens with probability $1/2$. On the other hand, if
	$|{\cal C}_{\mathcal{A}}(r(G^{i,l}_{k-1}) \cup r(G^{i,r}_{k-1}))|< dk/4$, then the addition of $R_i$ ensures that 
	$|{\cal C}_{\mathcal{A}}(r(G^i_k))| \geq dk/4$. Again, $R_i$ is added with probability $1/2$. In either case, given ${\cal E}^i$, 
	$\Pr[|{\cal C}_{\mathcal{A}}(r(G^i_k))| \geq dk/4]\geq 1/2$. We obtain
	$\Pr[|{\cal C}_{\mathcal{A}}(r(G^i_k))| \geq dk/4] \geq \Pr[|{\cal C}_{\mathcal{A}}(r(G^i_k))| \geq dk/4 \mid {\cal E}^i]\cdot \Pr[{\cal E}^i] 
	\geq \frac{1}{2}\cdot \frac{1}{4} = \frac{1}{8}$. Equivalently, $\Pr[|{\cal C}_{\mathcal{A}}(r(G^i_k))| < dk/4] \leq 7/8$. 
	If $|{\cal C}_{\mathcal{A}}(r(G_k))| < dk/4$, then $|{\cal C}_{\mathcal{A}}(r(G^i_k))| < dk/4$ must hold true for $i=1,\ldots, 6$. The 
	latter event occurs with probability at most $(7/8)^6$. We conclude $\Pr[|{\cal C}_{\mathcal{A}}(r(G_k))| \geq dk/4] \geq 1 -(7/8)^6 > 1/2$.
	This holds independently of $\mathcal{A}$'s coloring decisions made in other components. 
	
	Invariants (2--4) are immediate, based on the arguments given in the proof of Lemma~\ref{lem:det}. As for the number of
	vertices of any $G_k\in {\cal G}_k$, we observe that it is upper bounded by $12\cdot d \cdot (12^{k-1}-1) + 6 \cdot d/2 < d \cdot (12^k-1)$. 
	
	If $d$ is odd, the above construction of sets ${\cal G}_k$, $k\geq 1$, is performed for parameter $d-1$. In 
	${\cal G}_1$, graph $G_1$ is extended by a single vertex having edges to all other vertices in $G_1$. Invariant~(5) holds
	because any graph $G_k\in {\cal G}_k$ contains $12^{k-1}$ copies of $G_1$. 
	
	The lemma follows from (1) and (3--5). In particular,~(1) implies that the expected number of colors used by
	any deterministic online algorithm is at least $1/2 \cdot (d-1)k/4 = (d-1)k/8$. 
\end{proof}

\begin{proof}[Proof of Theorem~\ref{th:rand}]
	For the given $d$ and $n$, choose $k=\lfloor \log(n/d)\rfloor$. In this proof, logarithms are base~12. There holds 
	$k\in \mathbb{N}$, because $n\geq 12 d^2 > 12 d$. By Lemma~\ref{lem:rand}, there exists a probability
	distribution on a set ${\cal G}_k$ of chordal graphs with chromatic number $d$ such that the expected number of colors used
	by every deterministic online algorithm is at least $(d-1)k/8$. The number of vertices of any graph in ${\cal G}_k$
	is at most $d 12^k$. Hence, by the choice of $k$, it is upper bounded by $n$. For every $G_k\in {\cal G}_k$, we add a
	suitable number of vertices so that the total number of vertices is equal to $n$. Every new vertex has one edge to an 
	arbitrary vertex in the original graph $G_k$. Hence, there exists a probability distribution on a set of $n$-vertex 
	graphs with chromatic number $d$ such that the expected number of colors used by any deterministic online algorithm
	is at least $(d-1)k/8$. By Yao's principle~\cite{Y}, for every randomized online algorithm, there exists an $n$-vertex chordal
	graph $G$ with $\chi(G)=d$ such that the expected number of color is $c_k\geq (d-1)k/8\geq dk/16$. 
	We have $k\geq \log n-\log d-1 = \log(n/12)-\log d\geq 1/2\cdot \log(n/12)$, because $12 d^2 \leq n$, 
	and hence $d\leq \sqrt{n/12}$. Since $12 d^2 \leq n$, we have $\log(n/12) \geq 1/3\cdot \log n$ and
	thus $c_k \in \Omega(d\cdot \log n)$. 
\end{proof}
Again, in Theorem~\ref{th:rand} we can reduce the lower bound on $n$ from $12d^2$ to $12 d^{1+\epsilon}$, for any $0<\epsilon<1$. The expected number of colors used by $\mathcal{A}$ is $\Omega(\epsilon \cdot d \cdot \log n)$. 

\section{Further graph classes}\label{sec:classes}
Given Theorem~\ref{th:rand}, we can derive lower bounds on the performance of randomized online coloring algorithms
for other important graph classes.
\subsection{Trees, planar, bipartite, $d$-inductive and bounded-treewidth graphs}
\begin{corollary}\label{cor:tree}
	For every randomized online algorithm $\mathcal{A}$ and every $n \in \mathbb{N}$ with $n\geq 48$, there exists a $n$-vertex tree $T$, 
	presented by an oblivious adversary, such that the expected number of colors used by $\mathcal{A}$ to color $T$ is $\Omega(\log n)$. 
\end{corollary}
\begin{proof}
	The corollary follows from Lemma~\ref{lem:rand} and Theorem~\ref{th:rand} because, for $d=2$, the constructed graphs
	are trees. Every clique of size $d/2$ added in the construction is a singleton vertex. Indeed, every constructed graph is 
	a forest whose components can be linked by an additional vertex.
\end{proof}
Since trees are planar and bipartite graphs, we obtain the following two corollaries.
\begin{corollary}\label{cor:planar}
	For every randomized online algorithm $\mathcal{A}$ and every $n \in \mathbb{N}$ with $n\geq 48$, there exists a $n$-vertex planar graph $G$, 
	presented by an oblivious adversary, such that the expected number of colors used by $\mathcal{A}$ to color $G$ is $\Omega(\log n)$. 
\end{corollary}
\begin{corollary}\label{cor:bipartite}
	For every randomized online algorithm $\mathcal{A}$ and every $n \in \mathbb{N}$ with $n\geq 48$, there exists a $n$-vertex bipartite graph $G$, 
	presented by an oblivious adversary, such that the expected number of colors used by $\mathcal{A}$ to color $G$ is $\Omega(\log n)$. 
\end{corollary}
Every chordal graph $G$ is $(\chi(G)-1)$-inductive and has treewidth $\omega(G) -1 = \chi(G)-1$~\cite{Bo}. 
Hence, Theorem~\ref{th:rand} gives the following two results. 
\begin{corollary}\label{cor:dinductive}
	Let $d \in \mathbb{N}$ be an arbitrary positive integer. For every randomized online algorithm $\mathcal{A}$ and every
	$n \in \mathbb{N}$ with $n\geq 12d^2$, there exists a $n$-vertex $d$-inductive graph $G$, presented by an oblivious adversary, 
	such that the expected number of colors used by $\mathcal{A}$ to color $G$ is $\Omega(d\cdot \log n)$. 
\end{corollary}
\begin{corollary}\label{cor:treewidth}
	Let $d \in \mathbb{N}$ be an arbitrary positive integer. For every randomized online algorithm $\mathcal{A}$ and every
	$n \in \mathbb{N}$ with $n\geq 12d^2$, there exists a $n$-vertex graph $G$ of treewidth $d$, presented by an oblivious adversary, 
	such that the expected number of colors used by $\mathcal{A}$ to color $G$ is $\Omega(d\cdot \log n)$. 
\end{corollary}
The following corollary gives a result for strongly chordal graphs.
\begin{corollary}\label{cor:stronglychordal}
	Let $d \in \mathbb{N}$ be an arbitrary positive integer. For every randomized online algorithm $\mathcal{A}$ and every
	$n \in \mathbb{N}$ with $n\geq 12d^2$, there exists a $n$-vertex strongly chordal graph $G$ with chromatic number $\chi(G)=d$, presented by an oblivious adversary, 
	such that the expected number of colors used by $\mathcal{A}$ to color $G$ is $\Omega(d\cdot \log n)$. 
\end{corollary}
\begin{proof}
	We prove that every graph $G_k \in \mathcal{G}_k $ constructed in Lemma \ref{lem:rand} is strongly chordal. The corollary 
	then immediately follows from Theorem \ref{th:rand}.  Let $ N(v) $ denote the neighborhood of a vertex $ v $ in $ G_k $.
	For $ d \leq 3  $, $ G_k $ does not posses an even cycle and thus $ G_k $ is strongly chordal. For $ d \geq 4 $, consider an even 
	cycle $C$ of length at least six in $ G_k $. We first argue that there must exist two non-consecutive vertices $u$ and $v$ in $C$ that
	are part of the same tree node $w_T$. If $C$ visits only one or two tree nodes, this is obvious, given the length of $C$. If $C$
	visits at least three tree nodes, the desired fact follows from invariant~(2) in Lemma \ref{lem:rand}, which ensures that
	each connected component of $G_k$ forms a tree of tree nodes. 
	
	Hence let $u$ and $v$ be two non-consecutive vertices in $C$ 
	belonging to the same tree node $w_T$. As $w_T$ is a clique, $\{u,v\}$ is an edge in $G_k$. Moreover 
	$ N(u)\backslash \{v\} = N(v)\backslash \{u\} $ because, again, $w_T$ is a clique and its vertices are connected to the same 
	vertices that are not part of $ w_T $. Consider a neighbor $ s $ of $ v $ in $ C $. We differ between two cases. First, if $ s $ 
	is also a neighbor of $ u $  in $ C $, then there must exist a neighbor $ x $ of $ u $ and a neighbor $ y $ of $ v $ in $ C $ because 
	$ C $ has length at least six. Therefore, starting at $x$, cycle $C$ visits vertices $x, u, s, v, y$ in this order. 
	We have $ N(u)\backslash \{v\} = N(v)\backslash \{u\} $, which implies that $ \{u,y\}$ is an edge in $G_k$.  
        The distance between $u$ and $y$ in $ C $ is three so that $ \{u,y\} $ is an odd chord of $ C $. On the other hand, if $ s $ is not 
        a neighbor of $ u $ in $ C $, then $ \{u,s\}$ is an edge in $G_k$ since $ N(u)\backslash \{v\} = N(v)\backslash \{u\} $.
        Thus $ \{u,s\} $ is a chord of $ C $. Moreover $ \{u,v\} $ is a chord of $ C $ because $ u $ and $ v $ are non-consecutive in $ C $. 
        We conclude that either 
        $ \{u,v\} $ or $ \{s,u\} $ is an odd chord of $ C $ because distances between vertices $u$ and $v$ and vertices $s$ and $u$ in $C$ 
        differ by exactly one. 
\end{proof}

\subsection{Disk graphs}
A disk graph is the intersection graph of disks in the Euclidean plane. Every vertex corresponds to
a disk; two vertices are connected by an edge if the respective disks intersect. The following theorem implies that it is 
not possible to improve on the performance of deterministic online coloring algorithms by using
randomization. We use the common assumption that when an online algorithm makes coloring decisions, it does not use the disk 
representation~\cite{C+,E1,E2}.
\begin{theorem}\label{th:disk}
	Let $\mathcal{A}$ be an arbitrary randomized online algorithm. For every $n\in \mathbb{N}$ and $\rho\in \mathbb{R}$ with $\min\{n,\rho\}\geq 25$, 
	there exists a $n$-vertex disk graph $G$ with chromatic number $ \chi(G)=2 $, presented by an oblivious adversary, in which the ratio of the largest to smallest 
	disk radius is $\rho$, such that the expected number of colors used by $\mathcal{A}$ is $\Omega(\min\{\log n, \log \rho\})$.
\end{theorem}

\begin{figure}[ht]
	\begin{minipage}[t]{8.2cm}
		\begin{center}
			\includegraphics[width=0.50\textwidth]{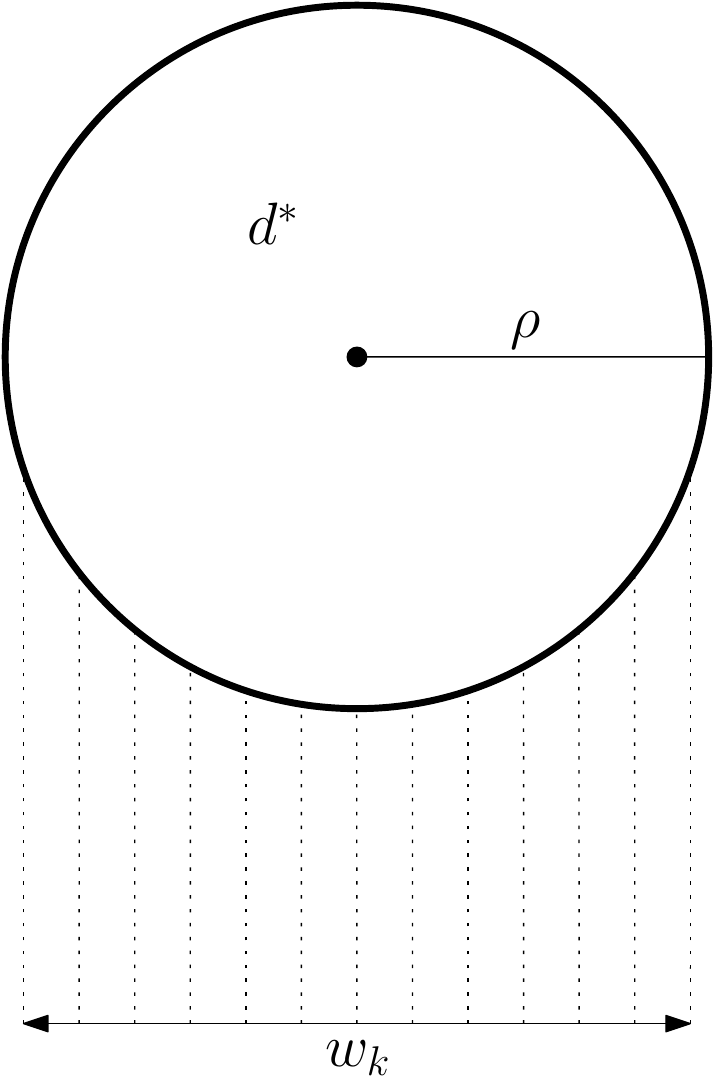}
		\end{center}
		\captionsetup{justification=centering}
		\caption{The disk graph $D_k$}\label{fig:disk1}
	\end{minipage}
	\begin{minipage}[t]{8.2cm}
		\begin{center}
			\includegraphics[width=0.3\textwidth]{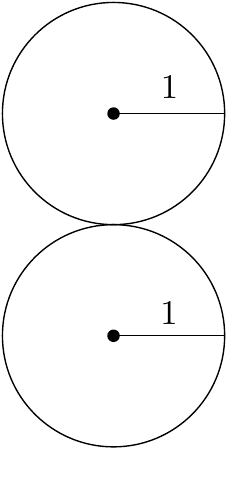}
		\end{center}
		\captionsetup{justification=centering}
		\caption{The disk graph $D_1$}\label{fig:disk2}
	\end{minipage}
\end{figure}

The proof of Theorem \ref{th:disk} relies on the following lemma, which we prove first.
\begin{lemma}\label{lem:disk}
	For every $k\in \mathbb{N}$ and every $\rho\in \mathbb{R}$ with $\rho> 12^{k-1}$, there exists a probability
	distribution on a set ${\cal D}_k$ of disk graphs with the following properties. In every $D_k\in {\cal D}_k$,
	the number of vertices is at most $2\cdot 12^k$, the ratio of the largest to smallest disk radius is $\rho$ and $ \chi(D_k)=2 $. 
	The expected number of colors used by every deterministic online algorithm for a graph drawn according to the
	distribution is at least $k/8$.
\end{lemma}
\begin{proof} 
	We use the graph sets ${\cal G}_k$, constructed in Lemma~\ref{lem:rand}, focusing on $d=2$. Let $k\in \mathbb{N}$
	be arbitrary. We show that if $\rho > 12^{k-1}$, every $G_k\in {\cal G}_k$ with $n_k$ vertices translates to a disk 
	graph $D_k$ with $n_k+1$ vertices in which the ratio of the largest to smallest disk radius is $\rho$. 
	
	Again $d=2$. Assume that $\rho>12^{k-1}$. Let $G_k\in {\cal G}_k$ be an arbitrary graph. The corresponding disk 
	graph is constructed in a top-down manner. In a first step we generate a graph $D_k$ in which disks touch each other 
	but do not intersect. At the end of the construction we slightly increase the disk radii so as to create intersections
	among the disks. 
	
	First in the construction of $D_k$ a disk $d^*$ of radius $\rho$, centered at the origin $(0,0)$, is placed in the
	Euclidean plane. Determine an $\epsilon>0$ such that $\rho > 12^{k-1}+\epsilon(12^{k-1}-1)12/11$. Such an $\epsilon$
	exists because $\rho > 12^{k-1}$. Disk $d^*$ has diameter $2\rho$. Let $S_k$ be the vertical strip of width
	$w_k= 2(12^{k-1}+\epsilon(12^{k-1}-1)12/11)$ below $d^*$. The center line of $S_k$ has $x$-coordinate~0. Figure~\ref{fig:disk1}
	depicts the arrangement. Since $2\rho>w_k$, disk $d^*$ extends past both boundaries of $S_k$. 
	In $S_k$ a disk representation of $G_k$ will be placed recursively below $d^*$. Recall that $G_k$ consists of twelve graphs 
	of ${\cal G}_{k-1}$ that form six pairs $(G_{k-1}^{i,l},G_{k-1}^{i,r})$, $1\leq i\leq 6$. To each pair a clique $R_i$ of 
	size $d/2$ might have been added. Since $d=2$, this clique is a single vertex.
	
	The disk representation of $G_k$ in $S_k$ is as follows. Strip $S_k$ of width $w_k$ is divided into twelve substrips of width 
	$w_k/12 = 2(12^{k-2}+\epsilon(12^{k-1}-1)/11)$ each, see again Figure~\ref{fig:disk1}. In the $m$-th substrip a strip 
	$S_{k-1}^m$ of width $w_{k-1}= w_k/12-2\epsilon = 2(12^{k-2}+\epsilon(12^{k-2}-1)12/11)$ is located, $1\leq m \leq 12$. 
	Strip $S_{k-1}^m$ lies in the middle of the $m$-th substrip so that its left and right boundaries have a distance
	of $\epsilon$ from those of the substrip. The strips $S_{k-1}^m$, $1\leq m \leq 12$, will host the representations
	of the twelve subgraphs of $G_{k}$. 
	
	Consider any pair $(G_{k-1}^{i,l},G_{k-1}^{i,r})$, $1\leq i\leq 6$. If $R_i$ is added to the pair, a disk $d_i$ of
	radius $r_{k-1}=w_{k-1}/2$ is placed in the odd numbered strip $S_{k-1}^{2i-1}$ below $d^*$. Disk $d_i$ is
	positioned so that it touches $d^*$. Then a representation of $G_{k-1}^{i,l}$ is placed in $S_{k-1}^{2i-1}$
	below $d_i$. Since $2r_k=w_k$, disk $d_i$ fully covers $S_{k-1}^{2i-1}$. If $R_i$ is not added to 
	$(G_{k-1}^{i,l},G_{k-1}^{i,r})$, a representation of $G_{k-1}^{i,l}$ is placed in $S_{k-1}^{2i-1}$ below 
	disk $d^*$. In any case a representation of $G_{k-1}^{i,r}$ is created in the neighboring strip $S_{k-1}^{2i}$
	below $d^*$. Since the left and right boundaries of $S_{k-1}^m$ have a distance of $\epsilon$ from those of 
	the $m$-th substrip containing $S_{k-1}^m$, disks and graph representations placed in $S_{k-1}^m$ do not 
	touch or overlap with disks placed in other strips $S_{k-1}^h$, $h\neq m$. 
	
	In general assume that a graph $G_j \in {\cal G}_j$, $k>j>2$, has to be represented in a strip $S_j$ of 
	width $w_j= 2(12^{j-1}+\epsilon(12^{j-1}-1)12/11)$ below a disk $d$. The construction proceeds in the same way as 
	described in the last paragraph for $j = k$. Strip $S_j$ is divided into twelve substrips, which in turn contain 
	strips of width $w_{j-1}= 2(12^{j-2}+\epsilon(12^{j-2}-1)12/11)$. These strips host the subgraphs of ${\cal G}_j$. 
	For each pair of subgraphs for which a new vertex is added, a disk of radius $r_{j-1}=w_{j-1}/2$ is placed so 
	that it touches disk $d$ from below. The top-down construction ends when graphs $G_1$ have to be placed in a 
	strip of width $w_1 = 2$ below a disk $d$. Graph $G_1$ is represented by a combination of two disks of 
	radius~1, see Figure~\ref{fig:disk2}. The two disks are placed on top of each other so that the upper one 
	touches $d$ from below. 
	
	Graph $D_k$ is constructed in a top-down manner. However, when presented to an online coloring algorithm, the 
	vertices are of course revealed bottom-up, with the vertices of graphs representing $G_1$ revealed first. 
	Disk $d^*$ in $D_k$ does not correspond to a vertex in $G_k$. For every other vertex of $G_k$, exactly one disk 
	was introduced. Hence, the analysis in the proof of Lemma~\ref{lem:rand} implies that $D_k$ contain at most  
	$1+2(12^k-1)\leq 2\cdot 12^k$ vertices. In $D_k$ the contact points among disks correspond to edges in $G_k$. 
	More precisely, any two disks $d$ and $d'$ touch each other in $D_k$ if and only if the vertices corresponding 
	to $d$ and $d'$ are connected by an edge in $G_k$. It is easy to modify $D_k$ so that the contact points are 
	replaced by real intersections among the respective disks. Let $\delta$ be the minimum distance of any disks 
	that do not touch each other. The radius of disk $d^*$ is increased by $\delta/2$. For every other disk the radius 
	is increased by $\delta/(2\rho)$.
	
	Let $D_k$ be the set of all disk graphs generated for any $G_k\in {\cal G}_k$. Consider the uniform distribution 
	on ${\cal D}_k$. By Lemma~\ref{lem:rand}, if a graph is drawn uniformly at random from ${\cal D}_k$, the expected number of 
	colors used by any deterministic online algorithm is at least $k/8$.
\end{proof}

\begin{proof}[Proof of Theorem~\ref{th:disk}]
	Let $k=\lfloor \log(\min\{n,\rho\}/2)\rfloor$. Logarithms are base~12. Since $\min\{n,\rho\}\geq 25$, there
	holds $k\in \mathbb{N}$. Moreover, $\rho> 12^{k-1}$. Consider the set ${\cal D}_k$ of disk graphs, defined in
	Lemma~\ref{lem:disk}, each of which consists of at most $2\cdot12^k$ disks. Hence by the choice of $k$, they
	consist of at most $n$ disks. To each $D_k \in {\cal D}_k$ with, say, $n_k$ disks we add $n-n_k$ additional
	disks. Their radius may be an arbitrary value between the smallest and the largest disk radius occurring in
	$D_k$. Lemma~\ref{lem:disk}, together with Yao's principle~\cite{Y}, implies that for every randomized 
	online algorithm there exists an $n$-vertex disk graph in which the ratio of the largest to smallest disk
	radius is $\rho$ such that the expected number of colors used by $\mathcal{A}$ is at least $k/8$. There holds
	$k\geq \log (\min\{n,\rho\}/2)-1 = \log (\min\{n,\rho\})-\log(24)\geq 1/100 \cdot \log (\min\{n,\rho\})$
	because $\min\{n,\rho\} \geq 25$. We conclude that $k/8\in \Omega(\min\{\log n, \log \rho\})$.
\end{proof}

\section{Lookahead and buffer reordering}\label{sec:look}
We explore the settings where an online algorithm has lookahead or is equipped with a reordering buffer.
\subsection{Lookahead} 
We first assume that a randomized online coloring algorithm $\mathcal{A}$ has lookahead~$l$. Theorem~\ref{th:look} below shows 
that, for chordal graphs, a lookahead of size $O(n/\log n)$ leads to no improvement.

\begin{theorem}\label{th:look}
	Let $d \in \mathbb{N}$ and $c\in\mathbb{R}$ be arbitrary numbers with $d\geq 2$ and $c\geq 1$. For every randomized online algorithm 
	$\mathcal{A}$ with lookahead $l$ and every $n \in \mathbb{N}$ with $n\geq \max\{12d^2, d\cdot 12^{2c}\}$ and $l\leq c n/\log(n/d)$,
	there exists a $n$-vertex chordal graph $G$ with chromatic number $\chi(G)=d$, presented by an oblivious adversary, 
	such that the expected number of colors used by $\mathcal{A}$ to color $G$ is $\Omega(\frac{1}{c} \cdot d\cdot \log n)$.
\end{theorem}
\begin{proof}
	For any $l \in \mathbb{N}$, consider the class of deterministic online algorithms with lookahead $l$. We refine the graph sets 
	${\cal G}_k$, defined in the proof of Lemma~\ref{lem:rand}, by specifying an order in which vertices arrive and by extending the 
	individual graphs. Let $k \in \mathbb{N}$ be arbitrary and $G_k \in {\cal G}_k$ be an any graph. The vertices of $G_k$ are 
	presented to an online coloring algorithm in phases, based on the subgraphs $G_j \in {\cal G}_j$ with $j<k$ contained in $G_k$. 
	
	More precisely, at the bottom level $G_k$ contains several instances of $G_1$. The vertices of all the copies of $G_1$
	form a set $P_1$. They are presented first and are part of phase~1. Next $G_k$ contains graphs $G_2\in {\cal G}_2$. Let $P_2$
	be the set of vertices in all instances of $G_2\in {\cal G}_2$ that are not yet contained in $P_1$. Using the notation of the
	proof of Lemma~\ref{lem:rand}, the vertices of $P_2$ belong to sets $R_i$ that are added to graph pairs of $G_1$. In general,
	let $P_j$ be the set of vertices in instances of $G_j \in {\cal G}_j$ that are not yet contained in $P_1\cup \ldots \cup P_{j-1}$,
	$1 <j\leq k$. Graph $G_k$ is presented to an online algorithm $\mathcal{A}$ by revealing the vertices of $P_j$, for increasing
	$j=1,\ldots, k$. The vertex sequence of $P_j$ forms phase~$j$, $1\leq j \leq k$.
	
	In order to render $\mathcal{A}$'s lookahead useless, at the end of each phase~$j$, exactly $l$ new dummy vertices are presented, 
	$1\leq j \leq k$. These new vertices can for example be isolated vertices. Alternatively, they could be combined to form a 
	chain of cliques having size at most $d$.  
	The new vertices increase the graph size by no more than $kl$. When coloring the 
	vertices of $P_j$, an online algorithm with lookahead~$l$ has no information about vertices of $P_{j'}$, $j'>j$.  
	Let ${\cal G}_k$ be the set of all extended graphs. Lemma~\ref{lem:rand} implies that if a graph is drawn 
	uniformly at  random from ${\cal G}_k$, the expected number of colors used by any deterministic online algorithm with 
	lookahead is at least $(d-1)k/8$.
	
	We conclude that for any $k,l \in \mathbb{N}$, there exists a probability distribution on a set $G_k$ of chordal graphs
	with the following properties. For every $G_k\in {\cal G}_k$, $\chi(G_k)=d$ and the number of vertices satisfies 
	$n_k \leq d\cdot 12^k +kl$. The expected number of colors used by any deterministic online algorithm with lookahead $l$ 
	is at least $(d-1)k/8$.
	
	In order to prove the theorem, for $d$, $c$ and $n$ with the stated properties, choose 
	$k=\lfloor \frac{2}{3c} \cdot \log (\frac{n}{d})\rfloor$. Logarithms are base~12.  There holds $k\in \mathbb{N}$ because 
	$n\geq d\cdot 12^{2c}$. Consider the set ${\cal G}_k$ of extended graphs defined above. Each graph in ${\cal G}_k$ has at 
	most $d\cdot 12^k +kl$ vertices. We argue that, for $l\leq c n/\log(n/d)$, this  expression is upper bounded by $n$. For 
	the chosen $k$, we have  $d\cdot 12^k \leq d^{1/3} n^{2/3} \leq n/3$. The first inequality holds because $c\geq 1$. 
	The second inequality is equivalent to $27\leq n/d$, which holds for $n\geq d\cdot 12^{2c}$. Obviously, if $l\leq c n/\log(n/d)$, 
	then $kl\leq 2n/3$. Hence, as desired,  $d\cdot 12^k +kl\leq n$. To each graph of ${\cal G}_k$ we add a suitable number of 
	vertices so that the graph size is exactly $n$.
	
	Using Yao's principle we obtain that, for every randomized online algorithm $\mathcal{A}$ with lookahead $l$, where $l \leq c n/\log(n/d)$, 
	there exists an $n$-vertex chordal graph $G$ with $\chi(G)=d$ such that the expected number of colors used by $\mathcal{A}$ is at least 
	$c_k\geq (d-1)k/8$, which in turn is lower bounded by $dk/16$. There holds $k\geq 2/(3c) \cdot \log (n/d) -1 = 
	2/(3c) \cdot (\log(n/d)) - (3c/2))$. Moreover, $\log(n/d)-(3c/2) \geq (1/4)\cdot \log(n/d)$, for $n\geq d\cdot 12^{2c}$. Also,  
	$\log (n/d) \geq {1/2}\cdot \log n$, for $n\geq d^2$. In conclusion, $k \geq 1/(12c) \cdot \log(n)$ and therefore $ c_k \in \Omega(\frac{1}{c} \cdot d \cdot \log(n)) $.
\end{proof}

Based on Theorem~\ref{th:look} we can derive analogous results for all the other graph classes considered in 
Section~\ref{sec:classes}. Loosely speaking, a lookahead of size $O(n/\log n)$ is of no help. The next 
Corollary~\ref{cor:tree2} addresses trees. Exactly the same statement holds for planar and bipartite graphs, 
respectively. For brevity, we omit the corresponding corollaries.
\begin{corollary}\label{cor:tree2}
	Let $c\geq 1$ be an arbitrary real number. For every randomized online algorithm 
	$\mathcal{A}$ with lookahead $l$ and every $n \in \mathbb{N}$ with $n\geq \max\{48, 2\cdot 12^{2c}\}$ and $l\leq c n/\log(n/2)$,
	there exists a $n$-vertex tree $G$, presented by an oblivious adversary, such that the expected number of colors 
	used by $\mathcal{A}$ to color $G$ is $\Omega(\frac{1}{c} \cdot \log n)$. 
\end{corollary}

For $d$-inductive graphs, graphs of treewidth $d$  and strongly chordal graphs with chromatic number $d$, the formulation of Theorem~\ref{th:look} directly carries over.
In fact, the result holds for all integers $d\geq 1$. 
For disk graphs, Theorems~\ref{th:disk} and \ref{th:look} give the following corollary. 
\begin{corollary}
	Let $c\in\mathbb{R}$ with $c\geq 1$ be arbitrary. For every randomized online algorithm $\mathcal{A}$ with lookahead $l$, every 
	$n \in \mathbb{N}$ and $\rho \in\mathbb{R}$ with $\min\{n,\rho \} \geq 2\cdot 12^{2c}$ and $l\leq c n/\log(n/2)$,
	there exists a $n$-vertex disk graph $G$ with chromatic number $ \chi(G)=2 $, presented by an oblivious adversary, in which the ratio of the largest to
	smallest disk radius is $\rho$, such that the expected number of colors used by $\mathcal{A}$ to color $G$ is 
	$\Omega(\frac{1}{c} \cdot \log n)$. 
\end{corollary}

\subsection{Buffer reordering} 
Next we examine the setting in which a deterministic online coloring algorithm $\mathcal{A}$ has a reordering buffer. We prove
that a buffer of size $n^{1-\epsilon}$, for any $0<\epsilon \leq 1$, does not improve the asymptotic performance of 
the algorithms.

\begin{theorem}\label{th:buffer}
	Let $d \in \mathbb{N}$ and $\epsilon\in\mathbb{R}$ be arbitrary numbers with $d\geq 2$ and $0< \epsilon \leq 1$. For 
	every deterministic online algorithm $\mathcal{A}$ having a buffer of size $b$ and every $n \in \mathbb{N}$ with $b\leq n^{1-\epsilon}$
	and $n\geq \max\{2d^2,2^{7/\epsilon}\}$, there exists a $n$-vertex chordal graph $G$ with chromatic number $\chi(G)=d$
	such that the number of colors used by $\mathcal{A}$ is $\Omega(\epsilon \cdot d\cdot \log n)$. 
\end{theorem}
\begin{proof}
	We extend the adaptive graph construction presented in the proof of Lemma~\ref{lem:det} and, as in the proof of 
	Theorem~\ref{th:look}, let the adversary generate a graph in a bottom-up fashion. Given $d$, $\epsilon$ and $n$ with the stated
	properties, let $k=\lfloor \log(n/d)\rfloor$. The adversary constructs a graph $G_k\in {\cal G}_k$ consisting of
	$2^{k-j}$ subgraphs $G_j\in {\cal G}_j$, for any $1\leq j\leq k$. In phase~1, $2^{k-1}$ graphs $G_1$ are constructed.
	As always each $G_1$ is a clique of size $d$ in which $d/2$ arbitrary vertices form a set of distinguished root vertices. 
	We assume that $d$ is even and address the case that $d$ is odd at the end of the proof. The vertices of all 
	the copies of $G_1$ may be presented in an arbitrary order to the deterministic online algorithm $\mathcal{A}$. In general 
	in phase $j$, $1<j\leq k$, the adversary presents the vertices of subgraphs $G_j \in {\cal G}_j$ that have not 
	been revealed in previous phases.
	
	More specifically, let $k'=\lfloor \frac{1}{2} (k-\log(4n^{1-\epsilon}/d))\rfloor$. There holds $k'\geq 1$: 
	The last inequality is satisfied if $\frac{1}{2} (\log(n/d)-1-\log(4n^{1-\epsilon}/d))-1\geq 1$.
	This inequality in turn is equivalent to $n\geq 2^{7/\epsilon}$, which holds true by the choice of $n$. Consider any phase 
	$j$, $1<j \leq k'$. We say that algorithm $\mathcal{A}$ \emph{has made progress on a subgraph $G_j\in {\cal G}_j$} if at the end of phase 
	$j$ the algorithm has colored at least half of the root vertices of $G_j$. We prove that at the end of every phase $j$, 
	$1\leq j\leq k'$, the following invariant~(1) holds. Invariants (2--5) are as in the proof of Lemma~\ref{lem:det}.  
	\\[-15pt]
	\begin{enumerate}[(1)]
		\item At the end of phase $j$, $\mathcal{A}$ has made progress on at least $2^{k-2j}$ subgraphs $G_j \in {\cal G}_j$. For each of these subgraphs, 
		$\left|\mathcal{C}_{A}\left(r(G_j) \right) \right| \geq \frac{d}{8}j$.
	\end{enumerate}
	For $j=1$, the analysis is simple. Suppose that at the end of phase~1 $\mathcal{A}$ has made progress on less than $2^{k-2}$ subgraphs 
	$G_1$. Then there exist more than $2^{k-1}-2^{k-2}=2^{k-2}$ graphs $G_1$ for which less than half of the root vertices have been 
	colored. Thus at the end of phase~1 the buffer must contain more than $\frac{d}{4} 2^{k-2}$ vertices. We observe that for any 
	$j$ with $1\leq j\leq k'$, there holds $\frac{d}{4} 2^{k-2j}\geq n^{1-\epsilon}$ because the latter inequality is equivalent to 
	$\frac{1}{2} (k-\log(4n^{1-\epsilon}/d))\geq j$, which is satisfied by the choice of $k'$. Since the buffer size is 
	at most $n^{1-\epsilon}$, $\mathcal{A}$ cannot store more than $\frac{d}{4} 2^{k-2}$ vertices in the buffer at the end of phase~1. Hence 
	$\mathcal{A}$ must have made progress on at least $2^{k-2}$ subgraphs $G_1$. For each of those subgraphs at least half of the root 
	vertices have been colored, i.e.\ at least $d/4>d/8$ colors have been used.
	
	Assume that invariant~(1) holds for phases $1,\ldots,j-1$, where $j\leq k'$. The adversary takes $2^{k-2(j-1)}$ subgraphs 
	$G_{j-1}\in {\cal G}_{j-1}$ for which $\mathcal{A}$ has made progress and $\left|\mathcal{C}_{\mathcal{A}}\left(r(G_{j-1}) \right) \right| \geq 
	\frac{d}{8}(j-1)$ holds. The adversary pairs them in an arbitrary way so that $2^{k-2j+1}$ graph pairs are formed. 
	Consider any such pair $(G_{j-1}^l, G_{j-1}^r)$. By inductive assumption $|\mathcal{C}_\mathcal{A}(r(G^l_{j-1}))| 
	\geq \frac{d}{8}(j-1)$ and $|\mathcal{C}_\mathcal{A}(r(G^r_{j-1}))| \geq \frac{d}{8}(j-1)$. If 
	$|\mathcal{C}_\mathcal{A}(r(G^l_{j-1}) \cup r(G^r_{j-1}))| \geq \frac{d}{8}j$, then the adversary creates a graph $G_j$ that
	is simply the union of $G_{j-1}^l$ and $G_{j-1}^r$. No further vertices are added. On the other hand, if 
	$|\mathcal{C}_\mathcal{A}(r(G^l_{j-1}) \cup r(G^r_{j-1}))| < \frac{d}{8}j$, the adversary creates a graph $G_j$
	by adding a clique $R$ of size $d/2$ to $(G_{j-1}^l, G_{j-1}^r)$. Each vertex of $R$ has an edge to every vertex of 
	$r(G_{j-1}^l)$. As in the proof of Lemma~\ref{lem:det} we can show that if $\mathcal{A}$ has colored at least half of the vertices
	of $R$, there holds $|\mathcal{C}_\mathcal{A}(r(G_j))| \geq \frac{d}{8} j$.  Phase $j$ consists of the arrival of the vertices of $R$, 
	taken over all the $2^{k-2j+1}$ graph pairs for which such a clique is added. Finally, the adversary takes the subgraphs 
	$G_{j-1}$ not combined so far and pairs them in an arbitrary way so as to create graphs $G_j$. No further vertices are added.
	
	It remains to verify that invariant (1) holds. Again, consider the $2^{k-2j+1}$ graph pairs composed of subgraphs 
	$G_{j-1}$ satisfying invariant~(1) for phase $j-1$. The adversary has constructed $2^{k-2j+1}$ corresponding subgraphs 
	$G_j$. We argue that at the end of phase $j$, $\mathcal{A}$ has made progress on at least half of them. Consider any $G_j$, 
	based on graph pair $(G_{j-1}^l, G_{j-1}^r)$. If no clique has been added, then $\mathcal{A}$ has made progress on $G_j$, 
	because by inductive assumption $\mathcal{A}$ has colored at least half of the root vertices of $G_{j-1}^l$ and $G_{j-1}^r$.
	On the other hand, if a clique $R$ had been added and $\mathcal{A}$ has not made progress on $G_j$, then more than $d/4$ vertices 
	of $R$ must reside in the buffer at the end of phase $j$. Hence, if $\mathcal{A}$ had not made progress on more than half of the 
	$2^{k-2j+1}$ considered subgraphs $G_j$, then more than $\frac{d}{4}\cdot \frac{1}{2} 2^{k-2j+1} = \frac{d}{4}2^{k-2j}$ vertices
	must reside in the buffer. This is impossible because, as verified in the second to last paragraph, $\frac{d}{4} 2^{k-2j} \geq
	n^{1-\epsilon}$. We obtain that $\mathcal{A}$ has made progress on at least $\frac{1}{2} 2^{k-2j+1}=2^{k-2j}$ subgraphs $G_j$.
	For each of these subgraphs, the potential addition of a clique $R$ ensures that $\mathcal{A}$ must use at least $\frac{d}{8} j$
	colors for the root vertices.
	
	After phase $k'$ the formation of graphs $G_j$, $k'<j\leq k$ is simple. The adversary takes arbitrary pairs of graphs 
	$G_{j-1}$ and combines them to form graphs $G_j$. No further vertices are added. Finally, when the generation of a graph 
	$G_k$ consisting of, say, $n_k$ vertices is complete, the adversary adds $n-n_k$ vertices to form a final graph $G$ 
	with $n$ vertices. Invariant~(1) for $j=k'$ ensures that $\mathcal{A}$ uses at least $\frac{d}{8} k'$ colors. 
	We show that $k'$ is in $\Omega( \epsilon \cdot \log n )$.
	There holds $k'\geq \frac{1}{2} (\log(n/d)-1-\log(4n^{1-\epsilon}/d))-1 = \frac{1}{2}(\epsilon \log n -5) \geq \frac{1}{8} \epsilon \log n$, for $n\geq 2^{7/\epsilon}$. 
	
	Finally, if $d$ is odd, the above graph construction is performed for $d-1$. A single vertex is added to each subgraph $G_1$
	to form a graph with clique size $d$.
\end{proof}
Given Theorem~\ref{th:buffer}, we derive analogous results for the other graph classes. Corollary~\ref{cor:tree3}
shows a result for trees. Identical statements hold for planar and bipartite graphs. Again, for brevity, we omit the corresponding corollaries.
\begin{corollary}\label{cor:tree3}
	Let $\epsilon\in\mathbb{R}$ with $0< \epsilon \leq 1$ be arbitrary. For every deterministic online algorithm $\mathcal{A}$ having a 
	buffer of size $b$ and every  $n \in \mathbb{N}$ with $b\leq n^{1-\epsilon}$ and $n\geq 2^{7/\epsilon}$, there exists a 
	$n$-vertex tree $G$ such that the number of colors used by $\mathcal{A}$ is $\Omega(\epsilon \cdot \log n)$. 
\end{corollary}

For $d$-inductive graphs, graphs of treewidth $d$ and strongly chordal graphs with chromatic number $d$, the statement of Theorem~\ref{th:buffer} directly carries over. 
In this case it holds for any $d\geq 1$. The corollaries are omitted here. Finally, we give a result for disk graphs.

\begin{corollary}\label{cor:disk3}
	Let $\mathcal{A}$ be an arbitrary deterministic online algorithm having a buffer of size $b$ and let $\epsilon\in\mathbb{R}$ be
	an arbitrary real number with $0< \epsilon \leq 1$. For every $n \in \mathbb{N}$ and $\rho\in \mathbb{R}$ with 
	$b\leq \min\{n^{1-\epsilon},\rho^{1-\epsilon}\}$ and $\min\{n,\rho\}\geq 2^{7/\epsilon}$, there exists a $n$-vertex 
	disk graph $G$ with chromatic number $ \chi(G)=2 $, in which the ratio of the largest to smallest disk radius is $\rho$, such that the number of colors 
	used by $\mathcal{A}$ is $\Omega(\epsilon \cdot \min\{\log n,\log \rho\})$.
\end{corollary}

\section*{Acknowledgments}
We thank anonymous referees for their valuable comments.

\bibliography{online_col_chordal_graph_new}

\end{document}